\documentclass[a4paper]{article}
\usepackage{fullpage}
\usepackage{amsmath}
\usepackage{amssymb}
\usepackage{url}
\usepackage[dvips]{graphicx}
\usepackage[ruled,vlined]{algorithm2e}
\usepackage{here}
\newtheorem{theorem}{Theorem}
\newtheorem{lemma}[theorem]{Lemma}
\newenvironment{proof}{\par\noindent{\sf Proof.}}{\par}
\newcommand{\QED}{\hfill\rule[1pt]{4pt}{6pt}}

\newcommand{\msize}[1]{{\left|#1\right|}}

\newcommand{\NP}{{\sf NP}}
\newcommand{\PP}{{\sf P}}

\newcommand{\calP}{\mathcal{P}}

\def\calI{{\cal I}}
\def\MPQ{{\cal MPQ}}
\def\P{{\cal P}}
\def\Q{{\cal Q}}

\begin{document}
\title{On Complexity of Flooding Games on Graphs with Interval Representations}

\author{Hiroyuki Fukui\footnotemark{}, 
Yota Otachi\footnotemark{},
Ryuhei Uehara\footnotemark{},
Takeaki Uno\footnotemark{}, and Yushi Uno\footnotemark{}}

\date{Japan Advanced Institute of Science and Technology (JAIST),
Nomi, Ishikawa 923-1292, Japan. \url{{s1010058,otachi,uehara}@jaist.ac.jp}\\
National Institute of Informatics (NII),
Chiyoda-ku, Tokyo 101-8430, Japan. \url{uno@nii.jp}\\
Osaka Prefecture University, 
Naka-ku, Sakai 599-8531, Japan. \url{uno@mi.s.osakafu-u.ac.jp}
}
\maketitle

\begin{abstract}
The flooding games, which are called Flood-It, Mad Virus, or HoneyBee,
are a kind of coloring games and they have been becoming popular online.
In these games, each player colors one specified cell in his/her turn, 
and all connected neighbor cells of the same color are also colored by the color. 
This flooding or coloring spreads on the same color cells.
It is natural to consider the coloring games on general graphs.
That is, once a vertex is colored, the flooding flows along edges in the graph.
Recently, computational complexities of the variants of 
the flooding games on several graph classes have been studied. 
In this paper, we investigate the flooding games 
on some graph classes characterized by interval representations.
Our results state that the number of colors is a key parameter to determine 
the computational complexity of the flooding games.
If the number of colors is not bounded,
the flooding game is \NP-complete even on caterpillars and proper interval graphs.
On the other hand, when the number of colors is a fixed constant,
the game can be solved in polynomial time on interval graphs.
We also state similar results for split graphs.

\end{abstract}

\section{Introduction}

The {\em flooding game} is played on a precolored board,
and a player colors a cell on the board in a turn.
When a cell is colored with the same color as its neighbor,
they will be merged into one colored area.
If a player changes the color of one of the cells belonging to 
a colored area of the same color, the color of all cells in the area are changed.
The game finishes when all cells are colored with one color.
The objective of the game is to minimize the number of turns
(or to finish the game within a given number of turns).
This one player flooding game is known as Flood-It (\figurename~\ref{fig:floodit}).
In Flood-It, each cell is a precolored square, the board consists of $n\times n$ cells,
the player always changes the color of the top-left corner cell, and 
the goal is to minimize the number of turns.
This game is also called Mad Virus played on a honeycomb board (\figurename~\ref{fig:MadVirus}).
One can play both the games online (Flood-It (\url{http://floodit.appspot.com/}) and
Mad Virus (\url{http://www.bubblebox.com/play/puzzle/539.htm})).


\begin{figure}[t]
\begin{minipage}{0.48\linewidth}
\begin{center}
\includegraphics[width=0.8\textwidth]
{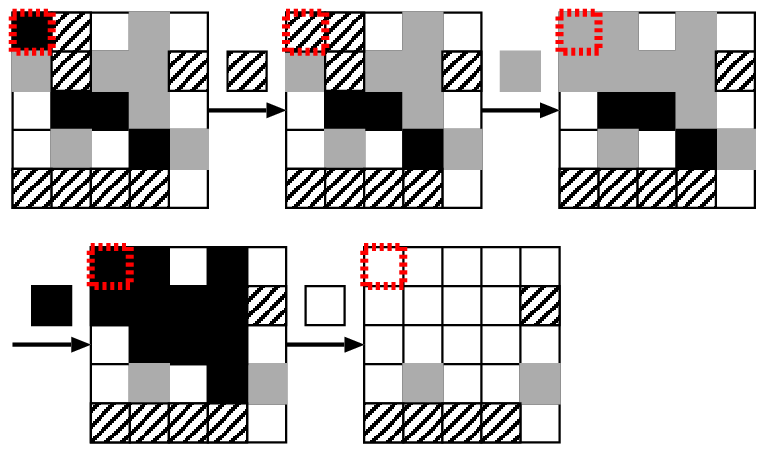}
\end{center}
\caption{A sequence of four moves on a $5\times 5$ Flood-It board.}
\label{fig:floodit}
\end{minipage}\hfill
\begin{minipage}{0.48\linewidth}
\begin{center}
\includegraphics[width=0.8\textwidth]
{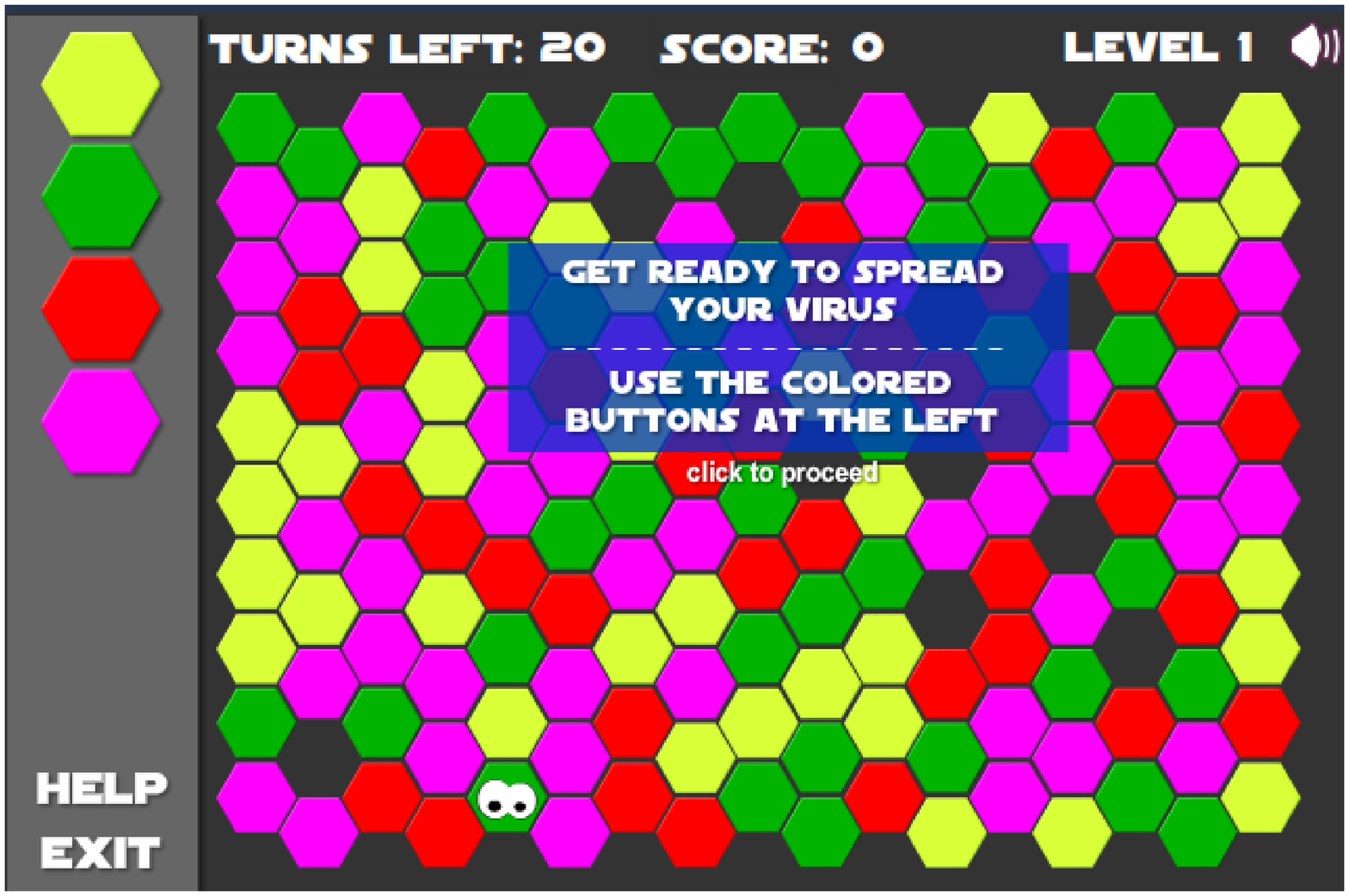}
\end{center}
\caption{The initial screen of the Mad Virus 
(\protect\url{http://www.bubblebox.com/play/puzzle/539.htm}).
The player changes the cell having eyes.}
\label{fig:MadVirus}
\end{minipage}
\end{figure}

In the original flooding games, the player colors a specified cell.
However, it is natural to allow the player to color any cell.
The original game is called {\em fixed} and this extended game is called {\em free}.
The flooding games are intractable in general on the grid board;
it is \NP-hard on rectangular $3\times n$ boards when 
the number of colors is 4 \cite{MeeksScott2012},
and it is still \NP-hard on rectangular $2\times n$ boards 
when the number of colors is $O(n)$ \cite{MeeksScott2011}.
On the other hand, Meeks and Scott also show an $O(h(k)n^{18})$ time
algorithm for the flooding game on $2\times n$ boards when the number of colors is $k$, 
where $h(k)$ is an explicit function of $k$ \cite{MeeksScott2011}.

In recent literature, the game board has been generalized to general graph;
that is, the vertex set corresponds to the set of precolored cells,
and two cells are neighbors if and only if 
the corresponding vertices are adjacent in the graph.
It is also natural to parameterize the number $k$ of colors.
The generalized flooding games on general graphs have been well investigated 
from the viewpoint of computational complexity.
We summarize recent results in \tablename~\ref{tab:summary}.
(The other related results can be found in 
\cite{MeeksScott2012a,CliffordJalseniusMontanaroSach2012}.)

\begin{table}
\begin{center}
\begin{tabular}{|c|cc|}\hline
Graph classes & fixed & fixed, $k$ is bounded \\\hline
general graphs 
 & \NP-C 
 & \NP-C if $k\ge 3$ \cite{ArthurCliffordJalseniusMontanaroSach2010}\\
 &       
 & \PP{} if $k\le2$ (trivial) \\\hline
($\square$/$\triangle$/hex.) grids
 & \NP-C & \NP-C if $k\ge 3$ \cite{LagoutteNaualThierry2011}\\\hline 
paths/cycles 
 & $O(n^2)$ \cite{LagoutteNaualThierry2011} 
 & $O(n^2)$ \cite{LagoutteNaualThierry2011} \\\hline
co-comparability graphs
 & \PP{} \cite{FleischerWoeginger2010}
 & \PP{} \cite{FleischerWoeginger2010} \\\hline
split graphs
 & \NP-C \cite{FleischerWoeginger2010}
 & \PP{} \cite{FleischerWoeginger2010} \\\hline
caterpillars
 & \PP{}\setcounter{footnote}{0}\footnotemark\addtocounter{footnote}{-1}
 & $O(4^k k^2 n^3)$ (This)\\\hline
proper interval graphs
 & \PP{}\setcounter{footnote}{0}\footnotemark\addtocounter{footnote}{-1}
 & $O(4^k k^2 n^3)$ (This)\\\hline
interval graphs
 & \PP{}\footnotemark
 & $O(4^k k^2 n^3)$ (This)\\\hline\hline
Graph classes & free & free, $k$ is bounded \\\hline
general graphs 
 & \NP-C  &  \NP-C if $k\ge 3$ \cite{ArthurCliffordJalseniusMontanaroSach2010} \\
 &        & \PP{} if $k\le2$ \cite{Lagoutte2010,LagoutteNaualThierry2011} \\\hline
($\square$/$\triangle$/hex.) grids
 & \NP-C & \NP-C if $k\ge 3$ \cite{LagoutteNaualThierry2011}\\\hline
paths/cycles 
 & $O(n^3)$ \cite{FukuiNakanishiUeharaUnoUno2011}\footnotemark\addtocounter{footnote}{-1}
 &  $O(n^3)$ \cite{FukuiNakanishiUeharaUnoUno2011}\footnotemark \\\hline
split graphs
 & \NP-C (This) 
 & $O((k!)^2+n)$ (This) \\\hline
caterpillars
 & \NP-C (This,\cite{MeeksScott2011})
 & $O(4^k k^2 n^3)$ (This)\\\hline
proper interval graphs
 & \NP-C (This)
 & $O(4^k k^2 n^3)$ (This)\\\hline
interval graphs
 & \NP-C (This)
 & $O(4^k k^2 n^3)$ (This)\\\hline
\end{tabular}
\end{center}
\caption{Computational complexities of the flooding games on some graph classes.}
\label{tab:summary}
\end{table}

Since the original game is played on a grid board, 
the extension to the graph classes having geometric representation is natural and reasonable. 
For example, each geometric object corresponding to a vertex can be regarded as 
a ``power'' or an ``influence range'' of the vertex.
That is, when a vertex is colored, 
the influence propagates according to the geometric representation.
Therefore, this game models epidemics, fires, and rumors on social networks.
In this paper, we first consider the case that the propagation is in one dimensional.
That is, we first investigate the computational complexities of
the flooding game on graphs that have interval representations.
We will show that even in this restricted case, 
the problem is already intractable in general.

From the viewpoint of the geometric representation of graphs,
the notion of interval graphs is a natural extension of paths.
(We here note that path also models rectangular $1\times n$ boards of the original game.)
A path is an interval graph such that each vertex has an influence on at most two neighbors.
In other words, each vertex has least influence to make the network connected.
In this case, the flooding games can be solved in polynomial time 
\cite{LagoutteNaualThierry2011,FukuiNakanishiUeharaUnoUno2011}.
However, we cannot extend the results for a path to an interval graph straightforwardly.
There are two differences between paths and interval graphs.
First, in an interval graph, short branches exist.
That is, one vertex can have three or more neighbors of degree one.
Second, interval graphs have twins;
two (or more) vertices are called twins if their (closed) neighbor sets coincide.
Interestingly, one of these two differences is sufficient to make 
the flooding game intractable:
\begin{theorem}
\label{th:NP}
The free flooding game is $\NP$-complete even on 
(1) proper interval graphs, and (2) caterpillars.
These results still hold even if the maximum degree of the graphs is bounded by 3.
\end{theorem}
Both of the classes of caterpillars and proper interval graphs 
consist of very simple interval graphs.
If the maximum degree is bounded by 2,
these classes degenerate to the set of paths.
Thus the results are tight.

General interval graphs have rich structure since 
vertices correspond to intervals of variant lengths. 
Therefore, it is not easy to solve the free flooding game efficiently.
However, when the number of colors is a constant, the game becomes tractable.
\begin{theorem}
\label{th:poly}
The free flooding game on an interval graph can be solved in $O(4^k k^2 n^3)$ time.
\end{theorem}
That is, the free flooding game on an interval graph 
is polynomial time solvable if the number $k$ of colors is fixed,
and that is $\NP$-complete if $k$ is not fixed.
Thus the game is fixed parameter tractable with respect to the number of colors.

We here compare our results with the results stated in \cite{MeeksScott2011}.
As mentioned above, a path of $n$ vertices is essentially the same as a
rectangular $1\times n$ board.
However, a $2\times n$ board cannot be modeled by an interval graph
since a $2\times 2$ board represents a cycle $C_4$, while $C_4$ is not an interval graph.
On the other hand, each vertex in a $2\times n$ board has degree at most three,
while the maximum degree of an interval graph is not bounded.
In a sense, the class of interval graphs is much larger than $2\times n$ boards;
for each $n$, there is only one $2\times n$ board, while there are exponentially many 
interval graphs with $2n$ vertices.
It is worth mentioning that our $O(4^k k^2 n^3)$ time algorithm is much faster, 
and is solvable for larger class than 
the $O(h(k)n^{18})$ time algorithm in \cite{MeeksScott2011}.

\setcounter{footnote}{1}
\footnotetext{The class of co-comparability graphs properly contains 
interval graphs and hence caterpillars and proper interval graphs.
Since this game is polynomial time solvable on a co-comparability graph, so they follow.}
\setcounter{footnote}{2}
\footnotetext{In \cite{FukuiNakanishiUeharaUnoUno2011}, 
the authors gave an $O(kn^3)$ algorithm. However, 
it can be improved to $O(n^3)$ easily in the same way in \cite{LagoutteNaualThierry2011}.}

We also extend the results for the fixed flooding game on a split graph
mentioned in \cite{FleischerWoeginger2010} to the free flooding game on a split graph.
Precisely, the free flooding game is $\NP$-complete even on a split graph,
and it can be solved in $O((k!)^2+n)$ time when the number $k$ of colors is fixed.

Although we only consider one player game in this paper, 
it is also natural to extend to multi-players.
Two-player variant is known as HoneyBee, which is available online at 
\url{http://www.ursulinen.asn-graz.ac.at/Bugs/htm/games/biene.htm}.
Fleischer and Woeginger have investigated this game from 
the viewpoint of computational complexity.
See \cite{FleischerWoeginger2010} for further details.

\section{Preliminaries}

We model the flooding game in the following graph-theoretic manner.
The game board is a connected, simple, loopless, undirected graph $G=(V,E)$.
We denote by $n$ and $m$ the number of vertices and edges, respectively.
There is a set $C=\{1,2,\ldots,k\}$ of colors, 
and every vertex $v\in V$ is precolored (as input) with some color $col(v)\in C$.
Note that we may have an edge $\{u,v\}\in E$ with $col(u)=col(v)$.
For a vertex set $U\subseteq V$, the vertex induced graph $G[U]$ is 
the graph $(U,F)$ with $F=E\cap\{\{u,v\}\mid u,v \in U\}$.
For a color $c\in C$, the subset $V_c$ contains all vertices in $V$ of color $c$.
For a vertex $v\in V$ and color $c\in C$, we define the {\em color-$c$-neighborhood}
$N_c(v)$ by the set of vertices in the same connected component as $v$ in $G[V_c]$.
Similarly, we denote by $N_c(W)=\cup_{w\in W}N_c(w)$ the color-$c$-neighborhood of 
a subset $W\subseteq V$.
For a given graph $G=(V,E)$ and the precoloring $col()$,
a {\em coloring operation} $(v,c)$ for $v\in V$ and $c\in C$ is defined by,
for each vertex $v'\in N_{c'}(v)\cup \{v\}$ with $c'=col(v)$,
setting $col(v')=c$.
For a given graph $G=(V,E)$ and a sequence 
$(v_1,c_1),(v_2,c_2),\ldots,(v_t,c_t)$ of coloring operations in $V\times C$,
we let $G_0=G$ and $G_i$ is the graph obtained by the coloring operation 
$(v_i,c_i)$ on $G_{i-1}$ for each $i=1,2,\ldots,t$.
In the case, we denote by $G_{i-1}\rightarrow_{(v_i,c_i)} G_{i}$
and $G_0\rightarrow^i G_{i}$ for each $0\le i\le t$.
Then the problem in this paper are defined as follows\footnote{In the fixed flooding game, 
$v_1=v_2=\cdots=v_t$ is also required.}:


\renewcommand{\algorithmcfname}{Problem}%
\begin{algorithm}[H]
 \caption{Free flooding game}
 \restylealgo{boxed}\linesnumbered
 \SetKwInOut{Input}{Input}
 \SetKwInOut{Output}{Output}
 \Input{A graph $G=(V,E)$ 
   such that each vertex in $V$ is precolored with $col(v)\in C$
   and an integer $t$;
 }
 \Output{Determine if there is a sequence of coloring operations
 $((v_1,c_1),(v_2,c_2),\ldots,(v_t,c_t))$ of length $t$ 
 such that all vertices in the resulting graph $G'$ (i.e.~$G\rightarrow^t G'$)
 have the same color;}
\end{algorithm}

For the problem, if a sequence of operations of length $t$ colors
the graph, the sequence is called a {\em solution} of length $t$.




\begin{figure}
\begin{center}
\includegraphics[width=0.8\textwidth]
{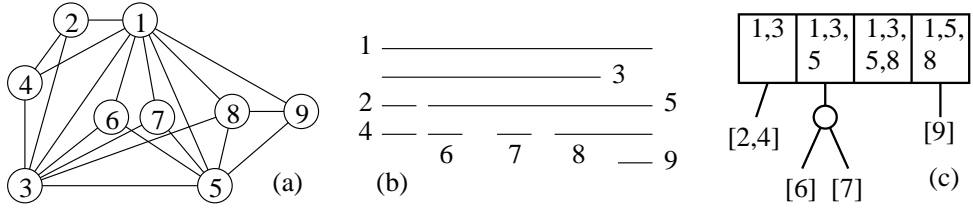}
\end{center}
\caption{(a) An interval graph $G$,
(b) one of interval representations of $G$, and
(c) unique $\MPQ$-tree of $G$ (up to isomorphism).}
\label{fig:interval}
\end{figure}

A graph $(V,E)$ with $V=\{v_1,v_2,\cdots,v_n\}$ is 
an {\em interval graph} if there is a set of (real) intervals
$\calI =\{I_{v_1},I_{v_2},\cdots,I_{v_n}\}$ such that
$\{v_i,v_j\}\in E$ if and only if $I_{v_i}\cap I_{v_j}\neq\emptyset$
for each $i$ and $j$ with $1\le i,j\le n$ (\figurename~\ref{fig:interval}(a)(b)).
We call the set $\calI$ of intervals an {\em interval representation} of the graph.
For each interval $I$, we denote by $L(I)$ and $R(I)$ the left and right
endpoints of the interval, respectively (hence we have $L(I)\le R(I)$ and $I=[L(I),R(I)]$). 
For a point $p$, let $N[p]$ denote the set of intervals containing the point $p$.
In general, there exist many interval representations for an interval graph $G$.
On the other hand, there exists unique representation for an interval graph $G$, 
which is called {\em $\MPQ$-tree} of $G$. The definition of $\MPQ$-tree is postponed to
Section \ref{sec:exint}.

An interval representation is {\em proper} if 
no two distinct intervals $I$ and $J$ exist such that $I$ properly contains $J$ or vice versa.
An interval graph is {\em proper} if 
it has a proper interval representation.
If an interval graph $G$ has an interval representation $\calI$ such that 
every interval in $\calI$ has the same length, $G$ is said to be a {\em unit} interval graph.
Such an interval representation is called a {\em unit} interval representation.
It is well known that the class of proper interval graphs 
coincides with the class of unit interval graphs \cite{Roberts1969}.
That is, given a proper interval representation, we can transform it into a unit interval
representation.
A simple constructive way of the transformation can be found in \cite{BogartWest1999}.
With perturbations if necessary, 
we can assume without loss of generality that $L(I)\neq L(J)$ (and hence $R(I)\neq R(J)$),
and $R(I)\neq L(J)$ for any two distinct intervals $I$ and $J$
in a unit interval representation $\calI$.

A connected graph $G=(V,E)$ is a {\em caterpillar} if $V$ can be
partitioned into $B$ and $H$ such that $G[B]$ is a path, and every
vertex in $H$ is incident to exactly one vertex in $B$. 
It is easy to see that the caterpillar $G$ is an interval graph.
We call $B$ (and $G[B]$) {\em backbone}, 
and each vertex in $H$ {\em hair} of $G$, respectively.

A graph $G=(V,E)$ is a {\em split graph} if $V$ can be partitioned into 
$C$ and $I$ such that $G[C]$ induces a clique and $G[I]$ induces an independent set.
(A vertex set $C$ is {\em clique} if every pair of vertices is joined by an edge,
and it is {\em independent set} if no pair is joined.)

\section{Graphs with interval representations}

Let $G=(V,E)$ be an interval graph precolored with at most $k$ colors.
We first show that, when $k$ is not bounded,
the flooding game on $G$ is \NP-complete even if 
$G$ is a caterpillar or a proper interval graph.
Next we show an algorithm that solves the flooding game in $O(4^k k^2 n^3)$
on a proper interval graph.
Lastly, we extend the algorithm to general interval graphs.
That is, the flooding game is fixed parameter tractable on 
an interval graph with respect to the number of colors.

\subsection{\NP-completeness on simple interval graphs}

To prove Theorem \ref{th:NP},
we reduce the following well-known \NP-complete problem to 
our problems (see \cite[GT1]{GJ79}):

\begin{algorithm}[H]
 \caption{Vertex Cover}
 \restylealgo{boxed}\linesnumbered
 \SetKwInOut{Input}{Input}
 \SetKwInOut{Output}{Output}
 \Input{A graph $G=(V,E)$ and an integer $k$;}
 \Output{Determine if there is a subset $S$ of $V$ such that
 for each edge $e=\{u,v\}\in E$, $e\cap S\neq \emptyset$ and $\msize{S}=k$;}
\end{algorithm}
Let $G=(V,E)$ and $k$ be an instance of the vertex cover problem.
Let $n=\msize{V}$, $m=\msize{E}$.

\begin{figure}
\begin{center}
\includegraphics[width=0.6\textwidth]{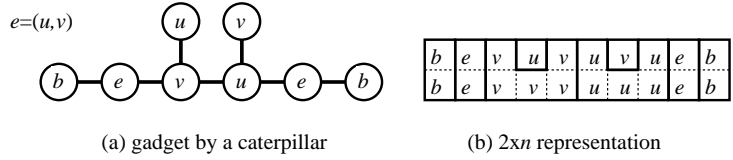}
\end{center}
\caption{A gadget for $e=(u,v)$.}
\label{fig:cat}
\end{figure}

\subsubsection{Caterpillar}
We first construct a caterpillar\footnote{We sometimes identify an
interval graph and its interval representation.}.
The key gadget is shown in \figurename~\ref{fig:cat}(a).
We replace an edge $e=(u,v)$ by a path 
$(b_1,b_2,b_3,b_4,b_5,b_6)$ with two hairs $h_3$ and $h_4$ attached to $b_3$ and $b_4$.
The colors are as shown in the figure.
Precisely, $col(b_1)=col(b_6)=b$, $col(b_2)=col(b_5)=e$,
$col(b_3)=col(h_4)=v$, and $col(b_4)=col(h_3)=u$.
It is not difficult to see that this gadget
cannot be colored in at most three turns.
On the other hand, there are some ways to color them
in four turns. One of them is:
color $b_3$ by $u$, color $b_3$ by $e$, color $b_3$ by $b$, and color $h_4$ by $b$.

\begin{figure}
\begin{center}
\includegraphics[width=\textwidth]{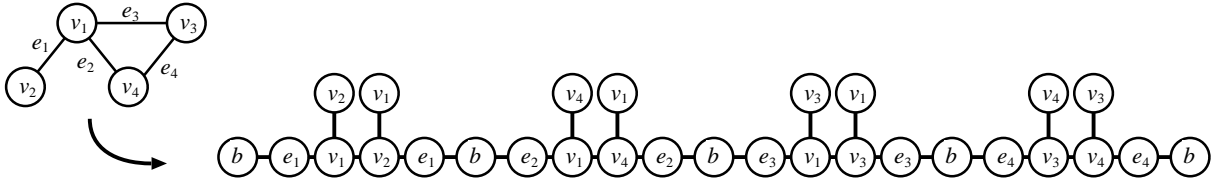}
\end{center}
\caption{An example of reduction to a caterpillar.}
\label{fig:cat2}
\end{figure}

Now we turn to the reduction of a general graph (\figurename~\ref{fig:cat2}).
We first arrange the edges in arbitrary way, and 
replace each edge by the gadget in \figurename~\ref{fig:cat}(a).
In this time, each vertex of color $b$ is shared by two consecutive edges.
In other words, endpoints of the gadget are shared by two consecutive gadget except both ends. 
This is the reduction. It is easy to see that the resultant graph is a caterpillar, 
the reduction is a polynomial-time reduction, and the flooding game is clearly in \NP.
Thus it is sufficient to show that 
a minimum vertex cover of $G$ gives a solution of the flooding game on the resultant graph 
and vice versa.

As shown in the example, all vertices on the backbone are colored by 
$b$ in $3m$ coloring operations.
On the other hand, $3m$ coloring operations are required to color the backbone.
Moreover, we have a leftover hair at each gadget, and their colors form 
a vertex cover $S$ since they hit all edges.
Therefore, once we have a vertex cover $S$, 
we color the resultant graph with $3m+\msize{S}$ operations.
On the other hand, if we can color the resultant graph with $3m+\msize{S'}$ operations,
we can extract $3m$ operations to color the backbone, and each of $\msize{S'}$ operations
is an operation to color a leftover hair, which gives us a vertex cover.
Therefore, the graph $G$ has a vertex cover of size $k'$ 
if and only if the resultant graph can be colored with $3m+k'$ coloring operations.
This completes the proof of Theorem~\ref{th:NP}(1).

We note that the basic idea of this reduction can be found in the proof
of the \NP-completeness on rectangular $2\times n$ boards in \cite{MeeksScott2011}.
In fact, the gadget in \figurename~\ref{fig:cat}(a) can be represented by 
a rectangular $2\times n$ board shown in \figurename~\ref{fig:cat}(b),
and we can obtain essentially the same proof in \cite{MeeksScott2011}.
We here explained the details of the proof to make this paper self-contained,
and this idea is extended to proper interval graphs in the next section.

\subsubsection{Proper interval graph}
We next construct an interval representation $\calI$ of 
a proper interval graph as follows (\figurename~\ref{fig:reduction}).

\begin{figure}
\begin{center}
\includegraphics[width=\textwidth]{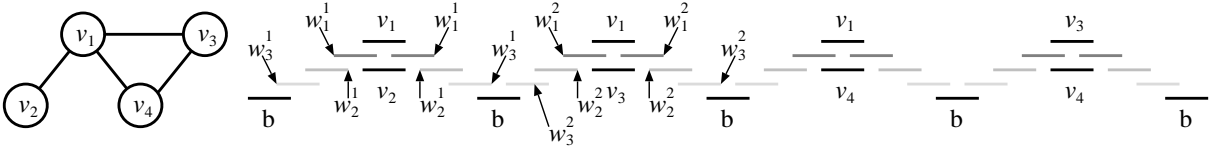}
\end{center}
\caption{Reduction from Vertex Cover to Flooding game.}
\label{fig:reduction}
\end{figure}
\begin{enumerate}
\item Let $C$ be the color set 
      $V\cup \{w_i^j \mid 1\le i\le m-1, 1\le j\le m\}\cup \{b\}$ of 
      $n+m(m-1)+1$ different colors.
      (Note that each vertex in $V$ has its own unique color.)
\item For each $0\le i\le m$, we put an interval $I_i=[4i,4i+1]$ 
      with precolor $col(I)=b$. We call these $m+1$ intervals {\em backbones}.
\item For each $e_i=\{u,v\}\in E$ with $0\le i < m$,
      we add two identical intervals 
      $J_i=[4i+2,4i+3]$ and $J'_i=[4i+2,4i+3]$ with precolor $col(J_i)=u$ and $col(J'_i)=v$.
      (Note that the ordering of the edges is arbitrary.)
\item Each two identical intervals $J_i$ and $J'_i$ are connected
      to the left and right backbone by paths of length $m$.
      Precisely, a left backbone $I=[4i,4i+1]$ and the two intervals
      $J_i=[4i+2,4i+3]$ and $J'_i=[4i+2,4i+3]$ are joined by 
      a path $(w_m^i,w_{m-1}^i,\ldots,w_1^i,w_0^i)$, 
      where $I=[4i,4i+1]=I_{w_m^i}$ and $J_i=I_{w_0^i}$ (which is identical to $J'_i$).
      (Note that $w_{1}^i$ has three neighbors: $w_{2}^i$
      and two vertices corresponding to $J_i$ and $J'_i$.)
      The intervals $J_i=[4i+2,4i+3]$, $J'_i=[4i+2,4i+3]$ are 
      connected to the right backbone $I=[4i+4,4i+5]$ in a symmetric way. 
      That is, they are connected by 
      a path $(w_{0}^i,w_{1}^i,\ldots,w_{m-1}^i,w_{m}^i)$ such that
      $I_{w_0^i}=[4i+2,4i+3]$, $I_{w_m^i}=[4i+4,4i+5]$.
      For each $j$ with $1\le j\le m-1$,
      we set $col(w_j^i)=w_j^i$ with $1\le j\le n$.
      That is, two paths from $[4i+2,4i+3]$ to both backbones have the same color sequence,
      and when we color the interval $J_i$ (or $J'_i$) by 
      the sequence $w_{1}^i,\ldots,w_{m-1}^i,b$, we can connect the left and right backbones.
\end{enumerate}
Now we show a lemma that immediately implies Theorem \ref{th:NP}(b).
\begin{lemma}
\label{lem:reduction}
In the reduction above, the original graph $G$ has a vertex cover of size $k'$
if and only if there is a sequence of coloring operations of length $m^2+k'$
to make the resulting interval representation in monochrome.
\end{lemma}
\begin{proof}
We first suppose that the graph $G$ has a vertex cover $S$ of size $k'$.
Then we can construct a sequence of coloring operations of length $m^2+k'$ as follows.
First step is joining the backbones.
Let $e_i=\{u,v\}$ be an edge in $E$. 
Since $S$ is a vertex cover, without loss of generality, we assume $u\in S$.
Then we color $v$ by $w_1^i, w_2^i, \ldots, w_{m-1}^i$, and $b$ 
(we do not mind if $v$ is in $S$ or not).
Repeat this process for every edge. 
Then all the backbones are connected and colored by $b$ after $m^2$ colorings.
We then still have $m$ intervals corresponding to the vertices in $S$.
Thus we pick up each vertex $v$ in $S$ and color the backbone by $col(v)$.
After $\msize{S}$ colorings, all vertices become monochrome.

Next we suppose that we have a sequence of coloring operations of length $m^2+k'$
that makes the representation monochrome. 
We extract a vertex cover of size $k'$ from these operations.
In the representation, for each $i$ with $1\le i\le m$,
we have $2$ distinct paths $(w_1^i,w_2^i,\ldots,w_{m-1}^i)$.
Hence we have $2m$ distinct paths in total, 
and each of them requires $m$ coloring operations.
Since $k'$ is the (potential) size of a vertex cover, 
we can assume that $k'<m$ without loss of generality.
First, we observe that the sequence of coloring operation 
includes $(v, w_1^i)$ or $(u, w_1^i)$ for each edge $e_i=(u,v)$.
Otherwise, we need more than $m$ coloring operations to connect 
the neighboring backbones (colored $b$). 
The operations never help to connect other backbones.
Thus the length of any sequence is no less than $m^2+m$.
Therefore, we can see either $(v, w_1^i)$ or $(u, w_1^i)$ appears in the sequence.
We say that $v$ is {\em selected} if $(v,w_1^i)$ appears
before $(u, w_1^i)$ (or $(u, w_1^i)$ may not appear).

Let $R$ be the set of vertices $v$ such that it is not selected in some edge $e_i=(v,u)$.
Then $R$ is a vertex cover.
Since the sequence makes all the vertices monochrome,
 the sequence includes either $(v,w_1^i)$ or $(*,v)$ 
 for each unselected vertex $v$ and edge $e_i=(v,u)$.
We call such operations {\em cover operations}.
Thus, the number of cover operations is no less than $\msize{R}$.
Remind that $m^2$ operations are needed to connect the selected intervals and paths, 
 and these operations are either
 of form $(*, w^i_j)$ or $(u, *)$ for selected $u$.
This implies that the length of the sequence is no less than
 $m^2+\msize{R}$, and thus $\msize{R} = k'$.
\QED\end{proof}

The reduction can be done in polynomial time. 
Hence, by Lemma \ref{lem:reduction}, Theorem \ref{th:NP}(2) immediately follows.

\subsection{Polynomial time algorithm on interval graphs for fixed color}

We first show an algorithm for proper interval graphs
that runs in polynomial time if the number of colors is fixed.
Next we extend the algorithm to deal with general interval graphs.

\subsubsection{Algorithm for a proper interval graph}
Let $\calI(G)$ be an interval representation of the proper interval graph $G=(V,E)$.
The interval representation is given in a compact form (see \cite{UeharaUno2007} for details).
Precisely, each endpoint is a positive integer, $N[p]\neq N[p+1]$ for each integer $p$,
and there are no indices $N[p]\subset N[p+1]$ or vice versa 
for each integer $p$ with $N[p]\neq\emptyset$ (otherwise we can shrink it).
Intuitively, each integer point corresponds to a set of different endpoints of the intervals
since the representation has no redundancy.
Then, it is known that $\calI(G)$ is unique up to 
isomorphism when $G$ is a proper interval graph (see \cite{SaitohYamanakaKiyomiUehara2010}),
and $\calI(G)$ can be placed in $[0..P]$ for some $P\le 2n-1$.
Sweeping a point $p$ from $0$ to $P$ on the representation, 
the color set $N[p]$ differs according to $p$.
More precisely, we obtain $2P+1$ different color sets 
for each $p=0, 0.5, 1, 1.5, 2, 2.5,\ldots,P-0.5, P$.
We note that for each integer $p$, $N[p+0.5]\subset N[p]$ and $N[p+0.5]\subset N[p+1]$.
Let $S_i$ be the color set obtained by the $i$th $p$ 
(to simplify the notation, we use from $S_0$ to $S_{2P}$).
Since the color set $C$ has size $k$, each $S_i$ consists of at most $k$ colors.
That is, the possible number of color sets is $2^k-1$ (since $S_i\neq\emptyset$).

Now we regard the unique interval representation as
 a path $\calP=(\hat{S_0},\hat{S_1},\ldots,\hat{S_{2P}})$,
 where each vertex $\hat{S_i}$ is precolored by the color set $S_i$.
Then we can use a dynamic programming technique,
 which is based on the similar idea to the algorithms for the flooding game
 on a path in \cite{FukuiNakanishiUeharaUnoUno2011,LagoutteNaualThierry2011}.
On a path, the correctness of the strategy comes from the fact 
 that removing the color at the point $p$ divides 
 the interval representation into left and right,
 and they are independent after removing the color at the point $p$.
However, on $\calP$, we have to take care of the influence of changing 
 the color set of a vertex in the original interval graph.
In the algorithms for an ordinary path,
 changing the color of a vertex has an influence to just two neighbors.
In our case, when we change a color $c$ in $S_i$ at a point $p$ to $c'$, 
all reachable color sets joined by $c$ from $\hat{S_{i}}$ are changed.
Thus we have to remove $c$ from $S_j$ and add $c'$ to $S_j$ for each $j$ with $i'\le j\le i''$,
where $i'$ and $i''$ are the leftmost and the rightmost vertices such that
$c\in \cap_{i'\le j\le i''} S_{j}$.
By this coloring operation, some colors may be left independent on the backbone of color $c'$. 
%
To deal with these color sets, 
we maintain a table $f(\ell,r,c,S)$ that is the minimum number of coloring operations to
satisfy the following conditions:
(1) $c\in S_{i}$ for each $i$ with $\ell\le i\le r$, and
(2) $\cup_{\ell\le i\le r}S_{i}\subseteq (S\cup \{c\})$.
That is, $f(\ell,r,c,S)$ gives the minimum number of coloring operations
to make this interval connected by the color $c$, and
the remaining colors in this interval are contained in $S$.
Once we obtain $f(0,2P,c,S)$ for all $c$ and $S$ on $\calP$,
we can obtain the solution by the following lemma:
\begin{lemma}
\label{lem:propercor}
For a given proper interval graph $G$, let $f(\ell,r,c,S)$ be the table defined above.
Then the minimum number of coloring operations to make $G$ monochrome is 
given by $\min_{c,S}(f(0,2P,c,S)+\msize{S})$.
\end{lemma}
\begin{proof}
We first show that we can make $G$ monochrome within 
$\min_{c,S}(f(0,2P,c,S)+\msize{S})$ coloring operations.
For each $c$ and $S$, by the definition of the table, 
we can make that every color set $S_i$ contains $c$ with $f(0,2P,c,S)$ coloring operations.
This means that every vertex $v$ is either $col(v)=c$ or $N(v)$ contains some $u$ such that 
$col(u)=c$, and $G[N_c(v)]$ is connected. Therefore, taking each color $c'\in S$,
and changing the color of any vertex of color $c$ to $c'$, all vertices of color $c'$ and $c$
are merged to the vertices of color $c'$.
Therefore, repeating this process, we can make $G$ monochrome with 
$\min_{c,S}(f(0,2P,c,S)+\msize{S})$ coloring operations.

We next show that the above strategy cannot be improved.
In the definition of the function $f$, we take a strategy that 
(1) first, color the interval $[i,j]$ with a color $c$ and 
(2) second, color the remaining colors in the interval $[i,j]$ by changing the color $c$.
We say that this color $c$ {\em dominates} the interval $[i,j]$ after the first step.
We suppose that a coloring operation pick up a color $c$ of a vertex $v$ and change it to 
another color $c'$. Then the color sets in an interval $[i,j]$ are changed since 
$i\le L(I_v)\le R(I_v) \le j$, 
$i$ is the leftmost vertex such that $S_i$ contains $c$, and 
$j$ is the rightmost vertex with $c\in S_j$.
Here we suppose that $[i,j]$ is properly contained in another 
interval $[i',j']$ with $i'\le i\le j\le j'$ that is dominated by a color $c''$.
Then, we can override to use the color $c''$ instead of $c$ in the sense that
changing $c$ to $c'$ is not better than changing $c''$ to $c'$.
That is, when we change a color $c$ to $c'$, if there is another overriding color $c''$,
it is not worse to change a color $c''$ to $c'$ instead of $c$.
Repeating this argument, we can see that the above strategy is not worse any other strategy.
Thus we cannot improve it.
\QED\end{proof}

This function satisfies the following recursive relation.
\[
\begin{array}{lcl}
\lefteqn{f(\ell,r,c,S) = \min\{ }\\
\min_{\ell<i\le r,c'\in C\setminus\{c\}} f(\ell, i-1, c, S')+f(i, r, c', S'')+1 &\mbox{such that} & S', S'' \subseteq S\cup \{c\}\\
\min_{\ell<i\le r,,c'\in C\setminus\{c\}} f(\ell, i-1, c', S')+1 + f(i, r, c, S'')&\mbox{such that}&S', S'' \subseteq S\cup \{c\}\\
\min_{\ell<i\le r} f(\ell, i-1, c, S')+f(i, r, c, S'') &\mbox{such that}&S', S'' \subseteq S\\
\}&&
\end{array}
\]
The correctness of this dynamic programming algorithm is given by Lemma \ref{lem:propercor}.
Thus the remaining task is showing the computational complexity of the function.
\begin{lemma}
\label{lem:algproper}
The value of $\min_{c,S}(f(0,2P,c,S)+\msize{S})$ can be computed in $O(4^k k^2 n^3)$ time.
\end{lemma}
\begin{proof}
This can be done in a standard dynamic programming technique.
Initialization step is that, for each $i$, 
$f(i,i,c,S)=0$ if $c\in S$ and $f(i,i,c,S)=1$ if $c\not\in S$ .
This step requires $(2P+1)\times k\times 2^k=O(2^k kn)$.
We also define $f(i,j,c,S)=0$ for any $j<i$ for convenience.

In general step, the algorithm computes $f(\ell,r,c,S)$ 
for each pair $\ell$ and $r$ with $\ell <r$. 
The algorithm computes all pairs $\ell$ and $r$ in 
the order $r-\ell=1$, $r-\ell=2$, $r-\ell=3$, $\ldots$, $r-\ell=2P$. 
For a pair $\ell,r$ with $\ell<r$, the algorithm next fix the color $c$. 
Then the algorithm generates all possible subsets $S$ of $C$. 
Using the above recursive relation, 
a value of $f(\ell,r,c,S)$ can be computed in 
$O(k(\ell-r)(2^{\msize{S}+1}+ 2^{\msize{S}+1}))$ time since
the values $f(\ell, i-1, c, S')$ and $f(i, r, c, S'')$ can be computed independently.
Therefore, in total, $f(0,2P,c,S)$ can be computed in 
$O(n^2\cdot k\cdot 2^k \cdot k\cdot n \cdot 2^{k})=O(4^k k^2 n^3)$ time.
This completes the proof.
\QED\end{proof}


In the proof, we assume that whether a color $c$ is in a color set $S$ or not can be 
determined in $O(1)$ time since $k$ is fixed.
Even in the case that $k=\msize{S}$ is large, say $O(n)$, 
the running time of the algorithm is bounded by $O(4^k k^2 (\log k) n^3 )$.


\subsubsection{Extension to interval graphs}
\label{sec:exint}

A proper interval graph has a simple interval representation.
Especially, its interval representation is linear and essentially unique up to isomorphism.
Therefore we can use the dynamic programming technique on the unique path-like structure.
On the other hand, a general interval graph has exponentially 
many different interval representations.
To deal with an interval graph, we use a tree representation 
that was used to solve the graph isomorphism problem for interval graphs \cite{KM89}.
The {\em $\MPQ$-tree} stands for {\em modified $\P\Q$-tree}, 
and this notion was introduced by Korte and M\"ohring in \cite{KM89}.
For an interval graph, the $\MPQ$-tree is uniquely determined up to isomorphism.
To solve the flooding problem on an interval graph,
we extend the algorithm for proper interval graph to solve 
the problem on the $\MPQ$-tree of color sets.
The $\MPQ$-tree maintains inclusion relationships among intervals.
That is, if an interval $I$ appears in a node $p$ that is an ancestor
of another node $q$, all intervals appearing in the node $q$ is properly contained in 
the interval $I$. Thus, by the same argument of the proof of Lemma \ref{lem:propercor},
a coloring operation for the intervals appearing in the node $q$ 
can be overridden by the coloring operation of $I$. 
Therefore, in the same manner of the algorithm for the color sets of proper interval graphs,
it is enough to consider the coloring operation of $I$, 
and the other intervals properly contained in $I$ will be dealt with as 
the set of remaining colors in $S$ in the function $f(0,2P,c,S)$ in the previous algorithm.
We discuss the details hereafter.

\paragraph{Definitions and Notations for {\MPQ}-trees:}
The notion of $\P\Q$-trees was introduced by Booth and Lueker \cite{BL76}.
A {\em $\P\Q$-tree} is a rooted tree $T$ with two types of 
internal nodes $\P$ and $\Q$, which will be represented by circles and rectangles, respectively.
The leaves of $T$ are labeled one-to-one with the maximal cliques
of the interval graph $G$. The {\em frontier} of a $\P\Q$-tree $T$ is the
permutation of the maximal cliques obtained by the ordering of the
leaves of $T$ from left to right. 
$\P\Q$-tree $T$ and $T'$ are {\em equivalent}, 
if one can be obtained from the other by applying the 
following rules a finite number of times;
\begin{description}
 \item[{\rm (1)}] arbitrarily permute the successor nodes of a $\P$-node, or
 \item[{\rm (2)}] reverse the order of the successor nodes of a $\Q$-node.
\end{description}
In \cite{BL76}, Booth and Lueker showed that a graph $G$ is an interval
graph if and only if there is a $\P\Q$-tree $T$ whose frontier represents
a consecutive arrangement of the maximal cliques of $G$.
%
In other words, if $G$ is an interval graph, all consecutive arrangements of the
maximal cliques of $G$ are obtained by taking equivalent $\P\Q$-trees.
%

The $\MPQ$-tree model, which stands for {\em modified $\P\Q$-tree}, 
is developed by Korte and M\"ohring to simplify the construction of 
a $\P\Q$-tree (see \figurename~\ref{fig:interval}(c) for an example). 
The $\MPQ$-tree $T^*$ assigns sets of vertices (possibly empty) to the
nodes of a $\P\Q$-tree $T$ representing an interval graph $G=(V,E)$.
A $\P$-node is assigned only one set, while a $\Q$-node has a set for each
of its sons (ordered from left to right according to the ordering of the sons).
For a $\P$-node $P$, this set consists of those vertices of $G$ contained
in all maximal cliques represented by the subtree of $P$ in $T$, but in no other cliques.

For a $\Q$-node $Q$, the definition is more involved.
Let $Q_1,\cdots,Q_m$ be the set of the sons (in consecutive order) of $Q$,
and let $T_i$ be the subtree of $T$ with root $Q_i$ (note that $m\ge 3$).
We then assign a set $S_i$, called {\em section}, to $Q$ for each $Q_i$.
Section $S_i$ contains all vertices that are contained in all maximal
cliques of $T_i$ and some other $T_j$, but not in any clique belonging
to some other subtree of $T$ that is not below $Q$.
The key property of $\MPQ$-trees is summarized as follows:
\begin{theorem}[{\cite[Theorem 2.1]{KM89}}]\label{th:MPQ}
Let $T$ be a $\P\Q$-tree for an interval graph $G=(V,E)$ 
and let $T^*$ be the associated $\MPQ$-tree. Let $\msize{V}=n$ and $\msize{E}=m$.
Then we have the following:
\begin{description}
\item[{\rm (a)}] $T^*$ can be obtained from $T$ in $O(n+m)$ time and represents $G$ in $O(n)$ space.
\item[{\rm (b)}] Each maximal clique of $G$ corresponds to a path in $T^*$ from the
 root to a leaf, where each vertex $v\in V$ is as close as possible to the root.
\item[{\rm (c)}] In $T^*$, each vertex $v$ appears in either one leaf, one $\P$-node, 
 or consecutive sections $S_i,S_{i+1},\cdots,S_{i+j}$ for some $\Q$-node with $j>0$.
\end{description}
\end{theorem}
Property (b) is the essential property of $\MPQ$-trees.
For example, the root of $T^*$ contains all vertices belonging to all
maximal cliques, and the leaves contain the simplicial vertices of $G$.
In \cite{KM89}, they did not state Theorem \ref{th:MPQ}(c) explicitly.
However, Theorem \ref{th:MPQ}(c) is immediately obtained from the fact
that the maximal cliques containing a fixed vertex occur consecutively in $T$.

In order to solve the graph isomorphism problem, a $\P\Q$-tree has
additional information which is called {\em characteristic node} in \cite{LB79,CB81}. 
This is the unique node which roots
the subtree whose leaves are exactly the cliques to which the vertex belongs. 
As noted in \cite[p.~212]{CB81}, the term characteristic node to mean the leaf, 
$\P$-node, or portion of a $\Q$-node which contains those cliques. 
Each vertex $v$ in $\MPQ$-tree directly 
corresponds to the characteristic node in the $\P\Q$-tree.
Although they did not discuss the uniqueness of $\MPQ$-tree in \cite{KM89},
their algorithm certainly constructs the unique $\MPQ$-tree for 
a given interval graph up to isomorphism \cite{Uehara2003}.


\paragraph{Algorithm for an interval graph:}
Let $G=(V,E)$ be a connected interval graph with $\msize{V}=n$ and $\msize{E}=n$
and $T^*$ the $\MPQ$-tree of $G$.
Let $U$ be the set of vertices in $G$ that appears in the root node of $T^*$.
\begin{lemma}
\label{lem:MPQ}
(1) $G[U]$ is a connected interval graph.
(2) The interval representation of $G[U]$ is unique up to isomorphism.
\end{lemma}
\begin{proof}
If $U$ is an empty set, $G$ is not connected (unless $G$ contains no vertex).
Hence $U$ contains at least one vertex. 
If the root is a $\P$-node, $G[U]$ is a clique, and its interval representation is unique.
Thus we consider the case that the root is a $\Q$-node.
By definition, this $\Q$-node corresponds to a unique interval representation of $U$.
If $G[U]$ is not connected, $G$ is also disconnected, 
which contradicts the assumption that $G$ is connected.
Thus we have the lemma.
\QED\end{proof}

A parent-child relationship on $T^*$ represents inclusion relationship.
That is, if a vertex $v$ is an ancestor of another vertex $u$ in $T^*$,
$I_v$ always contains $I_u$ in any interval representation for $G$.
Thus, for any interval representation of $G$, 
the union of the set of intervals corresponding to the vertices in $U$
contains all other intervals.
Moreover, any interval $I_u$ not in $U$ is properly contained in 
an interval $I_v$ in $U$.
Thus, using the same argument in the proof of Lemma \ref{lem:propercor},
we can claim that the coloring operation of $u$ is overridden by the coloring operation 
of $v$. That is, any sequence of coloring operations of intervals in an interval graph $G$
 can be overridden by a sequence of coloring operations that only consist of 
 the coloring operations of intervals in $U$.
Therefore, we can employ the same strategy for any interval graph $G$ as follows.
(1) First, color the intervals in $U$ and dominate the interval $[\min{}R,\max{}L]$
by some color $c$, 
where $\min{}R=\min\{R(I)\}$ and $\max{}L=\max\{L(I)\}$ for all intervals in $G[U]$.
We note that $G[U]$ has the unique interval representation by Lemma \ref{lem:MPQ}.
After (1), the intervals of color $c$ form the {\em backbone} of $G$ in the sense that
every vertex is either of color $c$ or adjacent to some vertices in the backbone.
(2) Second, color the remaining colors of the vertices by changing some vertex on the backbone.
Using the same arguments of the algorithm for proper interval graphs,
the correctness of the strategy follows.

Precise algorithm is as follows:
\begin{description}
\item[{\rm (a)}] For a given interval graph $G=(V,E)$, construct the $\MPQ$-tree $T$.
\item[{\rm (b)}] Pick up the root node of $T$, and let $U$ be the set of vertices
	   appearing in the root of $T$.
	   (Since $G$ is connected, we have $U\neq \emptyset$.)
\item[{\rm (c-1)}] 
 Case 1: $U$ is a $\P$-node. In this case, every vertex in $U$ is a universal vertex
	   that is adjacent to any other vertex.
	   Thus, pick up one vertex of $U$ and change color of it to each colors.
\item[{\rm (c-2)}] 
Case 2: $U$ is a $\Q$-node. 
Let $S_1,\ldots,S_k$ be the sections in $U$ in this ordering.
For each $i=1,\ldots,k$, 
let $S'_i$ be the set of colors of vertices in the subtree of the section $S_i$.
Then construct a path $\hat{P}=(\hat{S_1},\hat{S_2},\ldots,\hat{S_k})$, 
where $col(\hat{S_i})=col(S_i\cup S'_i)$. That is, each $\hat{S_i}$ is a color set
that consists of the colors of intervals in $S_i$ or $S'_i$.
Now, we apply the algorithm for the proper interval graph on the path $\hat{P}$.
\end{description}

Since the case (c-1) is an extreme case of the case (c-2), 
we concentrate on the case (c-2).
By Lemma \ref{lem:MPQ}, the interval representation for $U$ is unique,
and hence the sections are determined uniquely.
If we change the color of an interval $I$ through $S_i$,
it overrides all other intervals $I'$ in $S'_i$ since 
$I'$ appears in a descendant of $U$, which means that $I$ properly contains $I'$.
Therefore, without loss of generality, we can assume that a best strategy changes 
colors of intervals in $U$.
Therefore, all intervals $I'$ in $S'_i$ are dominated by any interval $I$ in $S_i$.
Hence we can concentrate to solve the problem on $G[U]$, and the other intervals 
contribute only as a color set in $G[U]$.
Thus we can apply the algorithm for the proper interval graph on $\hat{P}$, and 
obtain an optimal solution with the same time and space complexity in 
Lemma \ref{lem:algproper}, which completes the proof of Theorem \ref{th:poly}.


\section{Split Graphs}
In \cite{FleischerWoeginger2010}, the fixed flooding game on a split graph is investigated.
Using a similar idea in \cite{FukuiNakanishiUeharaUnoUno2011},
we can extend the results for the fixed flooding game to the free flooding game.
\begin{theorem}
\label{th:split}
(1) The free flooding game is $\NP$-complete even on a split graph.
(2) The free flooding game on a split graph can be solved in $O((k!)^2+n)$ time.
\end{theorem}
\begin{proof}
(1) In \cite{FleischerWoeginger2010}, the feedback vertex set problem
is reduced to the fixed flooding game on a split graph $G=(V,E)$.
The resulting graph $G$ consists of a clique $K$ and an independent set $I$.
Each vertex in $I$ has degree one except one universal vertex $u$ incident to 
all vertices in $K$. It is easy to see that this universal vertex $u$ can be 
one of the clique $K$. Now we add $\msize{K}$ vertices to $I$ and join them to $u$,
and each of them is colored by $\msize{K}$ colors that are same to the colors of 
vertices in $K$.
Then, the resultant graph is still split graph.
We consider the free flooding game on this new split graph.
Then, using the similar argument in \cite{FukuiNakanishiUeharaUnoUno2011},
this graph has a solution if and only if there is a sequence of operations that
always colors the universal vertex $u$.
Thus the feedback vertex set problem has a solution 
if and only if the free flooding game has a solution.

(2) We can observe that there is a solution of length at most $2k$ that first makes
all vertices in $K$ having the same color, and changes the color of the clique
to join the vertices in $I$. We can also see that there is an optimum solution of this form. 
This means that we always change the color of a clique vertex.
Since the vertices in $K$ of the same color are always connected,
the number of possibilities of each operation is at most $k'(k'-1)$, 
where $k'$ is the current number of colors used in $K$.
Thus, we can find an optimum solution in $O((k!)^2+n)$ time.
\QED\end{proof}

\section{Concluding remarks}
In this paper, we investigate the free flooding game on graphs that have interval representations. 
We show that this game is fixed parameter tractable with respect to the number of colors.
We also show the similar results for split graphs.
In \cite{FleischerWoeginger2010}, it is shown that 
the fixed flooding game on a co-comparability graph can be solved in polynomial 
time based on a dynamic programming technique.
In this case, computing a shortest path on a co-comparability graph can be 
a better idea than using the dynamic programming.
In the case, the idea may be extended to the free flooding game on a co-comparability graph, 
and we may obtain a polynomial time algorithm.

\section*{Acknowledgment}
The authors thank Eric Theirry for sending \cite{LagoutteNaualThierry2011}.


\bibliography{string,article8,articlee,article2012,book,myconf,mypaper}
\bibliographystyle{alpha}


\end{document}